\documentclass[12pt]{iopart}

\usepackage{iopams}

\usepackage{amsopn,amsthm,braket,enumitem}
\setenumerate{label=\upshape{(\roman*)}}
\newtheorem{thm}{Theorem}
\newtheorem{lem}[]{Lemma}
\newtheorem{prop}[]{Proposition}

\theoremstyle{definition}
\newtheorem{defi}[]{Definition}
\newtheorem{exam}[]{Example}
\newcommand{\lat}{\mathsf{L}} 
\newcommand{\vnl}{\mathsf{vNL}} 
\newcommand{\ep}{\varepsilon} 
\newcommand{\bR}{\mathbb{R}}
\newcommand{\ketbra}[1]{\ket{#1}\!\bra{#1}} 
\DeclareMathOperator{\lspan}{span}
\DeclareMathOperator{\cspan}{\overline{span}}
\DeclareMathOperator{\im}{im}

\newcommand{\RF}{Riesz-Fischer}
\newcommand{\seq}[2]{\paren{#1}_{#2}}
\newcommand{\seqi}[1]{\paren{#1}_{i\in I}}

\newcommand{\abs}[1]{{\left\vert #1 \right\vert}}
\newcommand{\paren}[1]{{\left( #1 \right)}}
\newcommand{\brac}[1]{{\left\{ #1 \right\}}}

\newcommand{\C}{{\mathbb C}}

\newcommand{\N}{{\mathbb N}}

\newcommand{\bZ}{{\mathbb Z}}
\newcommand{\Ref}[1]{(\ref{#1})}
\newcommand{\longto}{\longrightarrow}
\renewcommand{\a}{\alpha}
\newcommand{\om}{\omega}
\newcommand{\Om}{\Omega}
\newcommand{\q}{\quad}
\newcommand{\inr}[2]{\langle#1,#2\rangle}
\newcommand{\tx}{\text}
\let\tx\mbox

\renewcommand{\r}{\rho} 

\renewcommand{\Im}{{\rm Im}}
\renewcommand{\leq}{\leqslant}
\renewcommand{\geq}{\geqslant}

\begin{document}

 \title
 {Distinguishability of countable quantum states and von Neumann lattice}

 \author{Ry\^uitir\^o Kawakubo$^1$ and Tatsuhiko Koike$^{1,2}$}
 \address{$^1$ Department of Physics, Keio University, Yokohama 223-8522, Japan}
 \address{$^2$ REC for NS, Keio University, Yokohama 223-8521, Japan}
 \ead{rkawakub@rk.phys.keio.ac.jp and koike@phys.keio.ac.jp}
 \vspace{10pt}
 \begin{indented}
  \item[]\today
 \end{indented}

 \begin{abstract}
Condition for distinguishability
of countably infinite
number of pure states by a single measurement is given.
Distinguishability is to be understood as possibility of
an unambiguous measurement.
For finite number of states,
it is known that the necessary and sufficient condition
of distinguishability is that the states are linearly independent.
For infinite number of states, several natural classes of distinguishability can be defined.
We give a necessary and sufficient condition for a system of pure
states to be distinguishable.                                                                             
It turns out that each level of distinguishability naturally
corresponds to one of the generalizations of
linear independence to families of infinite vectors.
As an important example,
we apply the general theory to von~Neumann's lattice,
a subsystem of coherent states which corresponds to a lattice in the classical phase space.	
We prove that the condition for distinguishability is that
the area of the fundamental region of the lattice is greater than the Planck
constant, and also find subtle behavior on the threshold.
These facts reveal the measurement theoretical meaning of the Planck
constant and give a justification for the interpretation
that it is the smallest unit of area in the phase space.
The cases of uncountably many states and of mixed states are
also discussed.
 \end{abstract}

 \vspace{2pc}
 \noindent{\it Keywords}: 
 foundations of quantum mechanics, 
 measurement theory, 
 quantum information, 
 unambiguous measurement, 
 von~Neumann lattice, 
 Riesz-Fischer sequence

 \section{Introduction}

 Different states of a system are assumed to be distinguishable in classical mechanics. 
 This fundamental assumption, however, is abandoned in quantum mechanics.
 States cannot be distinguished without error unless they are orthogonal. 
 One would then like to consider a problem of distinguishing states in
 a given set, which is sometimes called state discrimination problem.

 Strategies for state discrimination can be classified into two types.
 In the first type, one makes a measurement with $n$ outcome in order to distinguish $n$ input states. 
 If a certain outcome is detected, we presume that the system is in the input state corresponding to that outcome. 
 This strategy is widely employed. 
 For example, Helstrom~\cite{Helstrom1967254} considered measurements of this type 
 to derive the minimum average error probability associated with the discrimination of arbitrary two states.
 In the second type, a measurement has $n+1$ outcomes. 
 The extra outcome corresponds to the answer ``do not know''.
 At the expense of this inconclusive outcome, 
 it is possible, under rather weak conditions on the input states,
 to distinguish the inputs \textit{with certainty} when the other $n$
 outcomes is obtained. 
 A measurement with this property is called an unambiguous
 measurement, which is the central subject of this paper. 
 Note that all measurements of the first type should be regarded as belonging to the second type,
 if we obliged to estimate the input even when we obtain the inconclusive outcome. 
 Consequently, unambiguous measurements cannot make the average error probability 
 smaller than that in measurements of the first type.
 In other words, unambiguous measurements are less suitable for quantitative studies of state discrimination.
 On the contrary, they are appropriate for qualitative studies, which we shall carry out in this paper.

 Unambiguous measurements for two pure states were discussed by Ivanovic~\cite{Ivanovic1987257} for the first time, 
 and later its theory was developed 
 by Dieks~\cite{Dieks1988303}, Peres~\cite{Peres198819}, Jaeger and Shimony~\cite{Jaeger199583}.
 Chefles~\cite{Chefles1998339} obtain a necessary and sufficient condition 
 that 
 a family consists of finitely many pure states can be unambiguously measured.
 The condition is the linear independence of the given family of states. 
 Sun et al.~\cite{PhysRevA.65.044306} and Eldar~\cite{1176618} 
 discuss optimal unambiguous measurements in relation to semidefinite programming.
 A necessary and sufficient condition for mixed states to be unambiguously measured was presented by Feng et al.~\cite{PhysRevA.70.012308}, 
 which is slightly complicated. 
 Note that all the studies above concern a family consisting of a finite number of states.

 Another subject of this paper is a von~Neumann lattice.
 A von~Neumann lattice is a family of states which corresponds to the
 lattice on the phase space in the classical mechanics. 
 This family is investigated in several contexts.
 von~Neumann~\cite{von1932mathematische} originally examines this family for simultaneous measurement of position and momentum.
 Gabor~\cite{5298517} discussed these families in the context of communication theory and electrical engineering, 
 which is a pioneering work in time-frequency analysis. 
 Interpolation problem for entire functions also has relation to von~Neumann lattices~\cite{Seip1992}.
 Properties of a von~Neumann lattice depend on the area of its fundamental region in the phase space. 
 Von~Neumann stated without proof that this family is complete when the area is roughly smaller than the Planck constant $h$.
 However, it was about 40 years later that 
 Perelomov~\cite{Perelomov:1971aa} and Bargmann et al.~\cite{Bargmann1971221} gave the proof for this fact. 
 Today, many of the properties have been revealed, which potentially offer measurement theoretical interpretations.

 In this paper, we investigate unambiguous measurements on countably many states.
 First, we develop a general theory of the distinguishability of countably many states.  
 We define a distinguishability of states as a possibility of unambiguous measurements on it.
 Then, we provide a necessary and sufficient condition for countable pure states to be distinguishable.
 We also consider uniform distinguishability, and give the maximum value of uniform success probability.
 We point out that there is a difference between distinguishability and
 uniform distinguishability in the case of infinitely many states. 
 Second, We apply the criterion of distinguishability developed in the first part of this paper to von~Neumann lattices.
 We find the measurement theoretical meaning of the Planck constant $h$, 
 the smallest unit of area in the classical phase space.  
 Depending on whether the spaces between the states is larger or smaller than the Planck constant, distinguishability of states changes drastically.

 This paper is organized as follows. 
 In Section~2, we define distinguishability of the states using unambiguous measurements.
 Then we briefly review properties of vectors in a Hilbert space, which can be considered
 as generalizations of liner independence, in Section~3.
 In Section~4, We show that distinguishability of countable pure states is equivalent to properties of vectors which we see in the previous section.
 Section~5 is devoted to investigations of von~Neumann lattices. 
 Conclusions and discussions are given in Section~6 and  Section~7.
 We discuss the case of uncountable states in \ref{counta}.

 \section{Distinguishability}
 \label{dist}

 We shall discuss the problem of distinguishing a quantum state in a given family by a single measurement. 
 We allow an answer ``do not know'' or ``unknown'', but do not allow taking one state for another in the family. 
 The problem is referred to as unambiguous measurement. 
 We shall consider an arbitrary quantum system described by a Hilbert space $H$. 
 Let $\seqi{\rho_i}$ be a given family of countable states, where $\rho_i$'s are  density operators on $H$. 
 In our terminology, countable includes finite.
 We will often denote a family $\seqi{\ketbra{\psi_i}}$ of pure states simply by $\seqi{\psi_i}$ in the following.

 A quantum measurement and the resulting probability density of
 the outcome is described by a POVM (e.g.~\cite[\S3.1]{davies1976quantum}). 
 A POVM $\Pi=(\Pi_j)_{j\in J}$ on $J$, where $J$ is a countable set, 
 is a list of bounded operators $\Pi_j$ on $H$ that satisfies the positivity, 
 $\Pi_j\geq 0$ for all $j\in J$, and the normalization, $\sum_{j\in J} \Pi_j =1$. 
 The sum should be understood in the sense of 
 the weak operator topology. 
 the conditional probability of obtaining an outcome $j \in J$ when 
 the input was $\rho_i$ is given by 
 \begin{eqnarray}
  q_{ji}(\Pi):=\tr[\Pi_j \rho_i]. 
 \end{eqnarray}
 When $I\subset J$, 
 the success probability of obtaining the outcome $i\in I\subset J$ for the input
 $i\in I$ is given by 
 \begin{eqnarray}
  q_i(\Pi):=q_{ii}(\Pi)=\tr[\Pi_i \rho_i].
 \end{eqnarray}
 Though these quantities depend on $\seqi{\rho_i}$, we omit them in the notation since we usually fix a family $(\rho_i)_{i\in I}$ in our discussion.

 We shall define the distinguishability of each state in a given family of states. 
 %
 %
 \begin{defi}[Distinguishability]
  \label{def-dist}
  Let $\seqi{\rho_i}$ be a countable family of states. 
  \begin{enumerate}
   \item 
                 \label{i-def-dist}
                 A POVM $\Pi=(\Pi_j)_{j\in J}$ 
                 {\em distinguishes}\/ (the states in) $\seqi{\r_i}$ 
                 if $J=I \sqcup \set{?}$, a disjoint union of $I$ and a set containing one element which we denote ``?'', and the following hold: 
                 \begin{eqnarray}
                  q_{ji}(\Pi)&=&0 \tx{ for all } i,j\in I, \q i\ne j, 
                   \\
                  q_i(\Pi)&>&0 \tx{ for all } i\in I.
                 \end{eqnarray}
   \item  
                  \label{i-def-unif}
                  A POVM $\Pi$ 
                  {\em uniformly distinguishes}\/ 
                  the states in $\seqi{\r_i}$ if $\Pi$ distinguishes  $\seqi{\r_i}$ and has constant $q_i(\Pi)$. 
                  The constant is called the {\em uniform success probability}. 
   \item 
                 \label{i-def-perf}
                 A POVM $\Pi$ 
                 {\em perfectly distinguishes}\/ 
                 the states in $\seqi{\r_i}$ if $\Pi$ distinguishes  $\seqi{\r_i}$ and $q_i(\Pi)=1$ for all $i\in I$.
  \end{enumerate}
  The family $(\rho_i)_{i\in I}$ of states is said {\em distinguishable}, {\em uniformly distinguishable}, or {\em perfectly distinguishable}\/
  if there exists a POVM that distinguishes, uniformly distinguishes, or perfectly distinguishes, respectively, the states in $\seqi{\r_i}$. 
 \end{defi}

 It is obvious by definition that perfect distinguishability implies uniform distinguishability, and that uniform distinguishability implies distinguishability.

 Uniform distinguishability allows another characterization: 
 \ref{i-def-unif}$'$
 There exists a POVM $\Pi$ such that it distinguishes the states $(\rho_i)_{i\in I}$ and $\inf_{i\in I}q_i(\Pi) >0$ holds. 
 Necessity is trivial. 
 For sufficiency, assume that 
 a POVM $\Pi'$ satisfies the condition \ref{i-def-unif}$'$. 
 Let $q_i':= q_i(\Pi')$ and $q':=\inf q_i' >0$.
 Then the POVM $\Pi=(\Pi_j)_{j\in I\sqcup \set{?}}$, where 
 \begin{eqnarray}
  \Pi_j&:=&
   \frac{q'}{q_i'}\Pi_j', \q j\in I, 
   \nonumber \\
   \Pi_?&:=&
   \Pi_?' + \sum_{i\in I}\left(1-\frac{q'}{q_i'}\right) \Pi_i', 
   \nonumber 
 \end{eqnarray}
 satisfies the condition \ref{i-def-unif}, because its success probabilities $q_i(\Pi)=q'>0$ do not depend on $i\in I$.

 Distinguishability and uniform distinguishability are equivalent 
 if the family $(\rho_i)_{i\in I}$ consists of only finite number of states,
 because the infimum of finitely many  positive numbers is positive.

 The conditions equivalent to distinguishability and perfect distinguishability 
 were discussed for finite number of states in~\cite{Chefles1998339}. 
 Uniform distinguishability, which is different from the condition 
 only when the family contains infinite number of states, is newly defined in this paper. 
 In the case that the number of states is countable, we will derive the necessary and sufficient condition 
 for each type of distinguishability, which is the theme of this paper. 
 The assumption of countability is not a restriction if the Hilbert space is separable. 
 See the \ref{counta} for details.

 It is worth noting that all conditions defined here are invariant under any unitary transformation, 
 in particular, under any unitary time evolution. 
 In fact, if a POVM $\Pi=(\Pi_j)_{j\in I\sqcup \set{?}}$ distinguishes the states $(\rho_i)_{i\in I}$, 
 then the POVM $\Pi'=(U\Pi_j U^*)_{j\in I\sqcup \set{?}}$ distinguishes the states $(U\rho_iU^*)_{i\in I}$, for any unitary operator $U$.

 \section{Properties of a family of vectors in a Hilbert space}
 \label{RF}

 In this section, we shall review some properties of a family of vectors in a Hilbert space,  
 which are related to the notion of linear independence. 
 Careful discussion is necessary if the number of vectors is infinite. 
 They will turn out to have substantial relation to distinguishability of quantum states in the next section. 
 Let $\lspan S$ be the minimal linear subspace containing $S$, where $S$ is a subset of $H$. 
 In other words, $\lspan S$ is the set of all linear combinations of finitely many elements of $S$. 
 Let $\cspan S$ be the norm closure of $\lspan S$. 
 %
 %
 \begin{defi}[e.g. {\cite[p.28]{young1980introduction}}, 
  {\cite[Definition~3.1.2, p.135]{christensen2003introduction}}]
  \label{def-indep}
  Let $\seqi{\psi_i}$ be 
  a family of vectors in $H$. 
  \begin{enumerate}
   \item[(n)] 
			  \label{i-indep}
			  The family $\seqi{\psi_i}$ is {\em linearly independent}\/ 
			  if $\psi_i \notin \lspan\set{\psi_j\in H | j\ne i, \, j\in I }$ for each $i\in I$.
   \item
                \label{i-minimal}
                The family $\seqi{\psi_i}$ is {\em minimal}\/
                if $\psi_i \notin \cspan\set{\psi_j\in H | j\ne i, \, j\in I }$ for each $i\in I$. 
   \item
                \label{i-RF}
                The family $\seqi{\psi_i}$ is {\em Riesz-Fischer}\/ 
                if there exists $A>0$ such that 
                \begin{eqnarray}
                 A \sum_{i \in I} |\alpha_i|^2 
                  \leq \left\| \sum_{i\in I} \alpha_i \psi_i \right\|^2 
                  \label{eq-RFcond}
                \end{eqnarray} 
                holds for scalars $\seqi{\alpha_i}$ with all but finitely many being zero. 
                We call the positive number $A$ a {\em Riesz-Fischer bound}.
   \item
                \label{i-orth}
                The family $\seqi{\psi_i}$ is {\em orthonormal}\/ 
                if $\inr{\psi_j}{\psi_i}=\delta_{ij}$ holds for all $i,j\in I$. 
  \end{enumerate}
 \end{defi}

 Orthonormality obviously implies the {\RF} property. 
 The {\RF} property implies minimality,  
 because otherwise there would be 
 $i\in I$ such that $\psi_i\in\cspan\brac{\psi_j|j\ne i}$, 
 and one could make 
 the norm of $\psi_i -\sum_{j\ne i} \a_j\psi_j$ smaller than 
 any given positive number by choosing $\a_j\in\C$ accordingly, 
 thus violating \Ref{eq-RFcond} for any $A>0$. 
 Minimality implies linear independence by definition.

 In the particular case that $I$ is a finite set, 
 linear independence, minimality and the {\RF} property are all equivalent. 
 Indeed, if $I$ is finite, then $\set{ (\alpha_i)\in \mathbb{C}^I | \sum_{i\in I} |\alpha_i|^2=1}$ is compact, 
 and we deduce {\RF} property from linear independence.

 In our discussion on distinguishability,
 the dual of a family $\seqi{\psi_i}$ of vectors in $H$ play a crucial role.
  Two lists $(\psi_i)_{i\in I}$ and $(\phi_i)_{i\in I}$ of vectors in $H$ are said 
  dual or {\em biorthogonal}\/ to each other, if $\inr{\phi_j}{\psi_i}=\delta_{ij}$ for all $i,j\in I$. 
  The condition for the existence of a dual is given by the following proposition.
  For the proof, see \cite[p.28]{young1980introduction}, \cite[Lemma
 3.3.1]{christensen2003introduction}. 
 We also attach a proof in \ref{proof-lembi} for the convenience of
 the  reader. 
  %
  %
  \begin{prop}
   \label{lembi}
   A family $(\psi_i)_{ i\in I}$ of vectors in $H$ admits a biorthogonal family if and only if it is minimal. 
   In that case, the biorthogonal family is unique
   if $(\psi_i)_{ i\in I}$ is complete i.e. $\cspan\set{\psi_i|i\in I}=H$. 
  \end{prop}
  The Riesz-Fischer property has a dual notion. 
  \begin{defi}
   The list $(\phi_i)_{ i\in I }$ of vectors in $H$
   is {\em Bessel}\/ 
   if there exists $B<\infty $, called a {\em Bessel bound}, such that 
   the following equivalent conditions are fulfilled.
   \begin{enumerate}[label=\upshape{\arabic*.}]
        \item
                           For scalars $\seqi{\alpha_i}$ with all but finitely many being zero,
                           \begin{eqnarray}
                                \left\| \sum_{i\in I} \alpha_i \phi_i \right\|^2
                                 \leq B\sum_{i \in I} |\alpha_i|^2 
                                 \label{eq-bessel-def1}
                           \end{eqnarray} 
                           holds.
        \item
                           For each vector $\psi \in H$,
                           \begin{eqnarray}
                                \sum |\braket{\psi,\phi_i}|^2 \leq B\|\psi\|^2
                                 \label{eq-bessel-def2}
                           \end{eqnarray} 
                           holds.
   \end{enumerate}
  \end{defi}
  For the proof of equivalence of two conditions,
  see \cite[Chapter~4, Section~2, Theorem~3]{young1980introduction}.

  The {\RF} property allows characterizations by the dual 
  ( {\cite[Chapter 4]{young1980introduction}},  
  see also {\cite[Proposition~2.3 (ii)]{casazza-christensen-li-lindner1.pdf}}).
  %
  %
  \begin{prop}
   \label{lemyoung}
   For a family $\seqi{\psi_i}$ of vectors and $A>0$, the following conditions are equivalent.
   \begin{enumerate}[label=\upshape{\arabic*.}]
        \item $(\psi_i)_{i\in I}$ is {\RF} with bound $A$.
        \item $(\psi_i)_{i\in I}$ 
          admits a biorthogonal family which is Bessel with bound $A^ {-1} $.
        \item 
                          For each 
                          $(\gamma_i)_{i\in I}\in \C^I$ with 
                          $\sum_{i\in I} |\gamma_i|^2 <\infty$, 
                          there exists $\phi\in H$ that satisfies 
                          $\braket{\phi,\psi_i}=\gamma_i$ and 
                          $\|\phi\|^2 \leq A^{-1}
                          {\sum_{i\in I} |\gamma_i|^2}$. 
   \end{enumerate}
  \end{prop}

  The moment problem is 
  to find 
  a vector $\phi\in H$ that satisfies $\braket{\phi,\psi_i}=\gamma_i$, 
  $i\in I$, 
  for a given family $(\psi_i)_{i\in  I}$ and a numerical sequence
  $(\gamma_i)_{i\in I} $,  
 In that context, the equivalence of (i) and (iii) states that the
  {\RF} property guarantees existence of a  solution of the moment problem for
  each square summable $(\gamma_i)_{i\in I} $.

 \section{Distinguishability of general family of pure states: the first main result} 
 \label{main}

 Now we shall state the condition for each type of distinguishability for a family $\seqi{\psi_i}$ of pure states. 
 %
 %
 \begin{thm}
  \label{thm1}
  Let $\seqi{\psi_i}$ be a family of countably many pure states. 
  \begin{enumerate}
   \item The states in $\seqi{\psi_i}$ are 
                 \label{i-dist=minimal}
                 distinguishable if and only if $\seqi{\psi_i}$ is minimal. 
   \item 
                 \label{i-unif=RF}
                 The states in $\seqi{\psi_i}$ are uniformly distinguishable with uniform success probability $q$
                 if and only if $\seqi{\psi_i}$ is Riesz-Fischer with bound $q$,
   \item
                \label{i-perf=orth}
                The states in $\seqi{\psi_i}$ are perfectly distinguishable 
                if and only if $\seqi{\psi_i}$ is an orthonormal family. 
  \end{enumerate}
 \end{thm}

 Only the \textit{if} part of (i) needs the assumption of countability.

 It should be noted that the statement \ref{i-unif=RF} of Theorem~\ref{thm1} reveals the significance of the Riesz-Fischer bound in 
 quantum measurement theory: success probability of uniform distinction.

 We remark that in the particular case of finite family $\seqi{\psi_i}$, 
 the statements \ref{i-dist=minimal} and \ref{i-unif=RF} of Theorem~\ref{thm1} are identical,  
 which reproduces the results obtained by Chefles~\cite{Chefles1998339}. 
 As was discussed in the previous sections, 
 linear independence, minimality and the {\RF} property are equivalent,  for finitely many vectors.

 We begin the proof of Theorem~\ref{thm1} with 
 %
 %
 \begin{lem}\label{lemele}
  If the POVM $\Pi=\seq{\Pi_j}{j\in I\sqcup\brac{?}}$ distinguishes the pure states in $(\psi_i)_{i\in I}$,
  then the following holds for $i,j,k\in I$.
   \begin{enumerate}[label=\upshape{\arabic*.}]
   \item 
                 $\Pi_j\psi_i=0$ for $i\ne j$. 
   \item 
                 $q_i(\Pi)=1$ if and only if $\Pi_?\psi_i=0$. 
   \item 
                 $\braket{\psi_i,\Pi_k\psi_j}= q_i(\Pi) \delta_{ik} \delta_{jk}$. 
  \end{enumerate}
 \end{lem}
 \begin{proof}
  The condition $\tr[\rho E] =0$ for $\rho=\ketbra{\psi}$ and $E\geq 0$ implies $E\rho=0$. 
  This fact together with the assumption of the lemma imply the first and second claims. 
  The third follows directly from the first. 
 \end{proof}

 \begin{proof}[Proof of Theorem~\ref{thm1}]
  We must prove the six claims below.
  
  1 (Distinguishability implies minimality). 
  Suppose $\Pi$ distinguishes $(\psi_i)_{i\in I}$. 
  Then it follows from Lemma~\ref{lemele} that 
  $(\phi_i)_{i\in I}$ defined by 
  \begin{eqnarray*}
   \phi_i&:=&\frac{\Pi_i\psi_i}{\braket{\psi_i,\Pi_i\psi_i}} 
  \end{eqnarray*}
  is biorthogonal to $(\psi_i)_{i\in I}$. 
  Thus $(\psi_i)_{i\in I}$ is minimal by Proposition~\ref{lembi}. 

  2 (Minimality implies distinguishability). 
  Suppose $(\psi_i)_{i\in I}$ is minimal. Then by Proposition
  \ref{lembi} 
  there exits a family $(\phi_i)_{i\in I}$ biorthogonal to $(\psi_i)_{i\in I}$.
  Define 
  \begin{eqnarray*}
   \Pi_i:=p_i \frac{\ketbra{\phi_i}}{ \| \,\ketbra{\phi_i} \,\|}, 
    \q 
    \Pi_?:=1-\sum_{i\in I} \Pi_i, 
  \end{eqnarray*}
  where $\sum_{i\in I} p_i=1$ and $p_i>0$ for all $i\in I$. 
  This is possible since $I$ is countable. 
  The operators $\Pi_j$ constitutes a POVM $\Pi=(\Pi_j)_{j\in I\sqcup\brac?}$ because positivity of $\Pi_?$ 
  is guaranteed by an inequality for the operator norm, 
  $ \left\| \sum_{i\in I} \Pi_i \right\| \leq \sum p_i = 1$.
  It follows from the biorthogonality that $\Pi$ distinguishes $(\psi_i)_{i\in I}$. 

  3 (Uniform distinguishability implies the Riesz-Fischer property).
  Suppose $\Pi$ uniformly distinguishes $\seqi{\psi_i}$. 
  Let $\psi=\sum_{i\in I} \alpha_i \psi_i$, 
  with all $\a_i$ but finitely many being zero. 
  One has 
  \begin{eqnarray*}
   \left\| \sum_i \a_i \psi_i\right\|^2
    =\braket{\psi,\psi}
    &=&\sum_{k\in I} \Braket{\psi,\Pi_k \psi} + \Braket{\psi,\Pi_?\psi}\\
   &\geq & \sum_{i,j,k\in I} \overline{\a_i} \braket{\psi_i,\Pi_k\psi_j}\a_j +0
    = q(\Pi) \sum_k |\a_k|^2, 
  \end{eqnarray*}
  where the inequality follows from positivity of $\Pi_?$  and the last equality from Lemma~\ref{lemele}. 
  Thus $\seqi{\psi_i}$ is {\RF} with bound $q(\Pi)$. 

  4 (The Riesz-Fischer property implies uniform distinguishability).
  Suppose $(\psi_i)_{i\in I}$ is Riesz-Fischer with bound $A$. By Proposition~\ref{lemyoung}, 
  there exists a biorthogonal family $(\phi_i)_{i\in I}$ which is Bessel with bound $A^{-1}$.
  It follows from \Ref{eq-bessel-def2} that 
  for any $\xi$ 
  the sum $\sum_{i\in I} |\braket{\phi_i,\xi }|^2$ converges. 
  This in particular implies that 
  only countably many $\braket{\phi_i,\xi }$ are nonzero. 
  With an appropriate order, 
  the sequence $\sum_{i\le k} |\braket{\phi_i,\xi }|^2$ becomes
  Cauchy. 
  From this fact and \Ref{eq-bessel-def1}, 
  one can show that
  \begin{eqnarray*}
   X:=\sum_{i\in I}\ketbra{\phi_i}
  \end{eqnarray*}
  converges in the strong operator topology 
  and defines a bounded liner map with bound $\|X\|\leq A^{-1}$.
  We define the POVM $\Pi=(\Pi_j)_{j\in I\sqcup\brac?}$ by 
  \begin{eqnarray*}
   \Pi_i:=\frac{\ketbra{\phi_i}}{\|X\|}, 
  \end{eqnarray*}
  where positivity of $\Pi_?$ is verified by 
  \begin{eqnarray*}
   \Pi_?:=1-\sum_{i\in I} \Pi_i =1 -\frac{X}{\|X\|}\geq 0.
  \end{eqnarray*}
  This POVM uniformly distinguishes $(\psi_i)_{i\in I}$ 
  with uniform success probability $q(\Pi)=1/\|X\|\geq 1/A^{-1} =A $.

  5 (Perfect distinguishability implies orthonormality).
  Let $\Pi$ distinguish $\seqi{\psi_i}$. 
  Then Lemma~\ref{lemele} implies that 
  \begin{eqnarray*}
   \braket{\psi_i,\psi_j}
    =\sum_{k\in I}
    \Braket{\psi_i,\Pi_k\psi_j} + \braket{\psi_i,\Pi_?\psi_j}
    =\delta_{ij} q_j(\Pi) +0
    =\delta_{ij}.
  \end{eqnarray*}
  Thus $\seqi{\psi_i}$ is an orthonormal family. 

  6 (Orthonormality implies perfect distinguishability). 
  If $\seqi{\psi_i}$ is an orthonormal family, the POVM $\Pi$ defined by 
  $\Pi_i:=\ketbra{\psi_i}$ and $\Pi_?:=1-\sum_{i\in I} \Pi_i$ perfectly distinguishes $(\psi_i)_{i\in I}$.
 \end{proof}

 If a family $\seqi{\psi_i}$ is given, one 
 interested not only 
 whether or not it is uniformly distinguishable 
 but also the value of the largest possible success probability. 
 The following theorem gives the value and the construction of 
 the measurement. 
 %
 %
 \begin{thm}\label{prop1}
  Suppose that a family $\seqi{\psi_i}$ of pure states is uniformly distinguishable. 
  By Proposition~\ref{lembi} and Theorem~\ref{thm1}, it  admits a unique biorthogonal family $\seqi{\phi_i}$ contained in $\cspan\set{\psi_i | i\in I} $.
  Then among all POVMs which uniformly distinguish $\seqi{\psi_i}$, 
  the POVM $\Pi=(\Pi_{j})_{j\in I\sqcup\brac?}$ defined by
  \begin{eqnarray}
   \Pi_j&:=&
        \cases{
        \frac{ \ketbra{\phi_j}}
        { \ \left\| \, \sum_{k\in I} \ketbra{\phi_k} \, \right\| \ },  
        &$j\in I$,
        \\ \\
   1-\sum_{k\in I} \Pi_k,  
        &$j=$?, 
        }
  \end{eqnarray}
  attains the maximum uniform success probability
  \begin{eqnarray}
   q(\Pi)= \frac{1}{ \ \left\| \, \sum_{k\in I} \ketbra{\phi_k} \, \right\| \  }.
  \end{eqnarray}
 \end{thm}

 Applying Theorem~\ref{prop1} to the case of two pure states
 $(\psi_i)_{i=1,2}$, 
 we can immediately obtain the maximum uniform 
 success probability 
 \begin{equation}
  q(\Pi)=1-\left|\braket{\psi_1|\psi_2}\right|, 
   \label{173747_27Jul15}
 \end{equation}
 by calculating the smallest nonzero eigenvalue of the operator
 $\sum_{i=1,2}\ketbra{\psi_i}$. 
 This is essentially the result of Dieks~\cite{Dieks1988303}, 
 though he gave a prior 
 uniform probability density $(p_i)_{i=1,2}=(1/2)_{i=1,2}$ to
 $(\psi_i)_{i=1,2}$ while we do not. 
 For finitely many states, uniform prior probability density is canonical in the sense of having a maximum entropy.
 However, in the case of infinity many states, uniform probability density does not exist. 
 This is the reason why we do not employ the prior probability in our
 discussion. 
 Our definition of uniform distinguishability is meaningful 
 even when the cardinality (number of elements) of $I$ is infinite.

 Prior to the proof for the theorem, we show a lemma 
 which states that any unambiguous measurement of a given family of pure states is essentially the projection to its biorthogonal family. 
 The lemma corresponds to equation (2.9) in \cite{Chefles1998339}, extended to the case of countable families $\seqi{\psi_i}$. 
 %
 %
 \begin{lem}\label{lemdual}
  Assume that a POVM $\Pi$ distinguishes a family $\seqi{\psi_i} $ of
  pure states
  so that, 
  by 
  Theorem~\ref{thm1} 
  and 
  Proposition~\ref{lembi}, 
  $\seqi{\psi_i} $ admits a unique biorthogonal family $\seqi{\phi_i}
  $ contained in $\cspan\set{\psi_i | i\in I} $. 
  Then $\Pi$ satisfies 
  \begin{eqnarray*}
   P\Pi_i P^*=q_i(\Pi)\ketbra{\phi_i}, 
  \end{eqnarray*}
  where $P$ is the orthogonal projection of $H$ onto $\cspan\set{\psi_i | i\in I}$. 
 \end{lem}
 \begin{proof}
  Assume for a while $\cspan\set{\psi_i |i\in I}=H$ and therefore $P=1$.
  For any $\psi=\sum \alpha_i \psi_i \in \lspan\brac{\psi_i|i\in I}$,
  it follows from Lemma~\ref{lemele} that 
  \begin{eqnarray*}
   \left\langle \psi, (\Pi_i-q_i(\Pi) \ketbra{\phi_i} ) \psi \right\rangle=0
  \end{eqnarray*}
  holds for all $i\in I$.
  Thus, by continuity of the inner product, $\Pi_i$, and $\ketbra{\phi_i}$ , 
  the equation above holds for all $\psi \in \cspan\set{\psi_i | i\in I}=H$. 
  One therefore has $\Pi_i=q_i(\Pi) \ketbra{\phi_i}$. 
  This proves the lemma in the case $\cspan\set{\psi_i | i\in I}=H$.
  When $\cspan\set{\psi_i | i\in I}\ne H$, define $\Pi'$ by 
  \begin{eqnarray*}
   \Pi_j':=
        \cases{
    P\Pi_j P^*, & $j\in I$,\\
   P\Pi_? P + (1-P), &$j=$?.  
        }
  \end{eqnarray*}
  Then $\Pi'$ is a POVM on the Hilbert space $PH$.
  The claim reduces to the case $\cspan\set{\psi_i | i\in I}=H$.
 \end{proof}

 \begin{proof}[Proof of Theorem~\ref{prop1}]
  As was shown in the proof of Theorem~\ref{thm1}, 
  $\|\sum_{i\in I} \ketbra{\phi_i}\|$ is finite 
  and the POVM $\Pi$ in the 
  theorem
  is well-defined. 
  Let $\Pi'$ be an arbitrary POVM which distinguishes
  $(\psi_i)_{i\in I}$ uniformly 
  and 
  let $P$ be the orthogonal projection from $H$ onto 
  $\cspan\set{\psi_i | i\in I}$. 
  It follows from 
  from Lemma~\ref{lemdual} and 
  $\Pi'_?\geq 0$ that 
  \begin{eqnarray*}
   P=P1P^*
    \geq \sum_{i\in I} P \Pi_i' P^*
    =q(\Pi')\sum_{i\in I} \ketbra{\phi_i}. 
  \end{eqnarray*}
  Thus one has $1 \geq \| P \| \geq q(\Pi') 
  \left\|\sum_{i\in I} \ketbra{\phi_i} \right\|
  =q(\Pi')/q(\Pi)$, i.e., $q(\Pi) \geq q(\Pi')$. 
 \end{proof}

 \section{Distinguishability of von Neumann lattices: the second main result} 

 In this section, we shall discuss distinguishability of coherent states 
 represented by a lattice in the complex plane, which is called the von~Neumann lattice. 
 The coherent states may be defined for a particle in one dimension or
 any quantum system which allows a harmonic oscillator, 
 so that they can represent photons, phonons, or other bosonic particles. 
 We do not specify the physical system here and treat them 
 in general, though we may sometimes use the terminology for a particle.

 Let $H$ be the Hilbert space which represents the states of a quantum
 system.   
 Let $a$ be the annihilation operator on $H$ that satisfies $aa^*-a^*a=1$. 
 Let $\ket{0}$ be a state which satisfies $a\ket{0}=0$.
 The state $\ket{0}$ is unique up to phase factor and is called the vacuum. 
 The coherent state $\ket z$, where $z\in \mathbb{C}$, is defined by \cite{sch26,PhysRev.131.2766}  
 \begin{equation}
  \ket{z}:=\exp[z a^* - \overline{z} a ]\ket{0}.
 \end{equation}
 They are minimum uncertainty states for the position operator $Q=2^{-1/2}(a+a^*)$ and the momentum operator $P=(2i)^{-1/2}(a-a^*)$. 
 This allows one to regard a coherent state $\ket{z}$ as the quantum state 
 that corresponds to the classical state represented by a single point $z=2^{-1/2}(q+ip)\in\mathbb{C}$ in the phase space. 
 It is easily verified by the equation $a\ket{z}=z\ket{z}$ that
 coherent states $(\ket{z})_{z\in\mathbb{C}}$ is linearly independent.
 It follows that mutually different finite number of states
 $(\ket{z_i})_{i=1,2,\dots,n}$ are uniformly distinguishable.

 In the context of simultaneous measurement of position and momentum, 
 von~Neumann considered the following family of coherent states, which
 corresponds to a lattice on the phase space. 
 %
 %
 \begin{defi}
  Let $\omega_1,\omega_2\in\C$ be such that $\Im(\omega_2/\omega_1)>0$. 
  Let 
  \begin{eqnarray}
   \lat(\omega_1,\omega_2)
        &:=&
        \{ n_1\omega_1+n_2\omega_2 \in \mathbb{C}\mid n_1,n_2 \in \mathbb{Z}
        \} 
        \subset\C, 
        \\
   \vnl(\omega_1,\omega_2)
        &:=&\{ \ket{\Omega} \mid \Omega\in \lat(\omega_1,\omega_2) \}
        \subset H. 
  \end{eqnarray}
  A family $\vnl(\omega_1,\omega_2)$ is called a von~Neumann lattice. 
  The set  $\set{t_1\omega_1 + t_2 \omega_2 | t_1,t_2\in [0,1) }\subset\C$ 
  is called the fundamental region of $\lat(\omega_1,\omega_2)$ or of
  $\vnl(\omega_1,\omega_2)$.  
  The area of the fundamental region is usually denoted by $S$.
 \end{defi}

 Though in some literature the name of von~Neumann lattice is used
 only for the case $S=\pi$,  
 we use the term for all $S$. 
 Von~Neumann lattices are called Weyl-Heisenberg systems in the
 field of time-frequency analysis or in the Gabor 
 analysis~\cite{5298517,Heil:2007aa}.  
 Von~Neumann stated without proof that a von~Neumann lattices is
 complete when $S\lesssim 1/4$  
 \cite[p.217, p.407 in the English version]{von1932mathematische}. 
 It was about 40 years later that Perelomov~\cite{Perelomov:1971aa} and  
 Bargmann et al.~\cite{Bargmann1971221} gave the proof for this fact. 
 Hereafter, we identify a set and a list of vectors 
 in discussion of von~Neumann lattices, 
 so that 
 we may say, for example,
 that a von~Neumann lattice is distinguishable, or 
 that it is linearly independent.

 We collect the properties of von~Neumann lattices~\cite{Perelomov:1971aa,Bargmann1971221,Seip1992} in Table~\ref{table}. 
 These properties of von~Neumann lattices are essentially attributed to, 
 via the Fock-Bargmann space~\cite{CPA:CPA3160140303}, the nature of
 entire functions in complex analysis. 

 \begin{table}[h!]
  \begin{center}
   \begin{tabular}{l| l| c| cc}
    ~       & state families           & complete        & minimal        & Riesz-Fischer  \\ \hline
    $S<\pi$ & $\vnl(\omega_1,\omega_2)$                          &Yes &No &No \\ 
        & $\vnl^{(n)}(\omega_1,\omega_2),\ \ 1\leq n<\infty$  &Yes &No &No \\ \hline
    $S=\pi$ & $\vnl(\omega_1,\omega_2)$                          &Yes &No &No \\
    ~       & $\vnl^{(1)}(\omega_1,\omega_2)$                    &Yes &Yes &No${}^\dagger$ \\
    ~       & $\vnl^{(n)}(\omega_1,\omega_2),\ \ 2\leq n<\infty$  &No &Yes &No${}^\ddagger$ \\ \hline
    $S>\pi$ & $\vnl(\omega_1,\omega_2)$                          &No &Yes &Yes \\
   \end{tabular}
  \end{center}
  \caption{
  Properties of von~Neumann lattices :
  $\vnl^{(n)}$ denotes the set obtained by removing arbitrary $n$ elements from $\vnl$. 
  Facts marked by $\dagger$ and $\ddagger$ can be deduced from (38) in \cite{Perelomov:1971aa}.
  However, they have not been stated explicitly. 
  } 
  \label{table}
 \end{table}

 Combining the properties of von~Neumann lattices and 
 Theorem~\ref{thm1},  
 we arrive at our second main result. 
 %
 %
 \begin{thm}
  \label{thm-vnl}
  Let $S$ be the area of the fundamental region of $\vnl(\omega_1,\omega_2)$, 
  Then, the followings hold.
  \begin{enumerate}[labelindent=2pc,style=multiline,leftmargin=5pc]
   \item
                When $S< \pi$, $\vnl(\omega_1,\omega_2)$ is not distinguishable.
   \item
                When $S= \pi$, $\vnl(\omega_1,\omega_2)$ is not distinguishable.       
                However, the set $\vnl(\omega_1,\omega_2)$ with more
                than one element removed  
                is distinguishable, but is not uniformly distinguishable. 
   \item
                When $S> \pi$, $\vnl(\omega_1,\omega_2)$ is uniformly
                distinguishable, but is not perfectly distinguishable. 
  \end{enumerate} 
 \end{thm}

 A finite subset of coherent states is always uniformly distinguishable,
 whereas von~Neumann's lattice, which is infinite, behaves quite different.
 The result shows that each level of distinguishability is directly
 related to the area $S=\Im\left(\overline{\omega_1}\,\omega_2\right)$ 
 and does not depend on $\omega_1$ and $\omega_2$ separately. 
 Furthermore, it can be shown that
 distinguishability of von~Neumann's lattice is determined solely by the density of points in the complex plane, $1/S$, and is robust to deformation of the lattice (See proofs in \cite{Perelomov:1971aa} and \cite{Seip1992}).
 The threshold $S=\pi$ 
 corresponds to the area of the Planck constant $h$ in the
 classical phase space 
 [note that $d^2z=(h/2\pi)^{-1}(2^{-1/2} dq)(2^{-1/2}dp)$]. 
 Physically, the area $h$ is the minimum unit of area of the phase space, 
 which appeared e.g. 
 in the Bohr-Sommerfeld quantum condition~\cite[(48.2)]{landau1977quantum}.
 Theorem~\ref{thm-vnl} 
 reveals the measurement theoretical meaning of the Planck constant. 

 Finally, we give a simple and direct estimate for the uniform success
 probability, though it is not tight. 
 %
 %
 \begin{thm}\label{prop2}
  $|\omega_1|,|\omega_2|\to \infty$ with $\sin\arg(\overline{\omega_1}\,\omega_2)>\ep$ for some $\ep>0$
  forces maximal uniform success probability $q(\Pi)$ of $\vnl(\omega_1,\omega_2)$ approaches $1$.
  More specifically, when $\omega_1,\omega_2\in\mathbb{C}$ satisfies $S=\Im(\overline{\omega_1}\,\omega_2)> \pi$, 
  there exists a POVM $\Pi$ which distinguishes $\vnl(\omega_1,\omega_2)$ uniformly with uniform success probability
  \begin{eqnarray*}
   q(\Pi) 
        \geq 2- 
        \left(
         1+
         \frac{2\sqrt{\pi}}{\sin\arg(\overline{\omega_1}\,\omega_2)\min_i|\omega_i|}
        \right)^2.
  \end{eqnarray*}

 \end{thm}
  \begin{proof}
   Let $\Omega_{nm}=n\omega_1+m\omega_2\in\lat(\omega_1,\omega_2)$.
   For any numerical sequence $(\alpha_{n,m})_{(n,m)\in \mathbb{Z}^2}$
   with only finitely many being nonzero, one has 
   \begin{eqnarray*}
        &&\left\|\sum \a_{m,n} \ket{\Omega_{m,n}} \right\|^2 \\
        &=&\sum_{m,n} \sum_{k,\ell}  \overline{ \a_{k,\ell} }\a_{k+m,\ell+n} 
         \braket{\Omega_{k,\ell}|\Omega_{k+m,\ell+n}} \\
        &=&\sum_{k,\ell} |\a_{k,\ell}|^2
         +\sum_{(m,n)\ne(0,0)}  \sum_{k,\ell}
     \overline{\a_{k,\ell}}\a_{k+m,\ell+n}  \braket{\Omega_{k,\ell}|\Omega_{k+m,\ell+n}} .
   \end{eqnarray*}
   The second sum in the last line can be estimated as 
   \begin{eqnarray*}
        &&\left| \sum_{(m,n)\ne(0,0)}  \sum_{k,\ell} 
           \overline{\a_{k,\ell}}\a_{k+m,\ell+n}  \braket{\Omega_{k,\ell}|\Omega_{k+m,\ell+n}}  \right|\\
        &\leq & 
         \sum_{(m,n)\ne(0,0)}    
         e^{ -|\Omega_{m,n}|^2/2 }
         \sum_{k,\ell}  |\overline{\a_{k,\ell}}\a_{k+m,\ell+n}|\\
        &\leq & 
         \sum_{(m,n)\ne(0,0)}  
         e^{ -|\Omega_{m,n}|^2/2 }
         \sum_{k,\ell}  |\a_{k,\ell}|^2, 
   \end{eqnarray*}
   where we have used the formula $|\braket{z|w}|^2=e^{-|z-w|^2}$ 
   and the triangle inequality in the first line 
   and the Cauchy-Schwarz inequality in the second. 
   Therefore, we have
   \begin{eqnarray*}
        \left\|\sum \a_{m,n} \ket{\Omega_{m,n}} \right\|^2
         \geq 
         A
         \sum_{k,\ell}  |\a_{k,\ell}|^2,
         \qquad A := 1 -\sum_{(m,n)\ne(0,0)}  
         e^{ -|\Omega_{m,n}|^2/2 }. 
   \end{eqnarray*}
   In the case $A>0$, $\vnl(\omega_1,\omega_2)$ is Riesz-Fischer with bound $A$.
   So, by the Theorem~\ref{thm1}, $\vnl(\omega_1,\omega_2)$ is uniformly distinguishable 
   with uniform success probability at least $A$.

   Since 
   $
   \left|\overline{\lambda} \lambda+\overline{\lambda}\nu\right|
   \geq \left|\Im\left( \overline{\lambda}\lambda+\overline{\lambda}\nu\right)\right|
   =\left|\Im\left(\overline{\lambda}\nu\right)\right|
   $
   , we have 
   \begin{eqnarray*}
        |\Omega_{m,n}|
         &=& 
         \frac{1}{|\omega_1m|} 
         \left| \overline{\omega_1m}\,\omega_1m+\overline{\omega_1m}\, \omega_2n \right|\\
        &\geq &
         \frac{1}{|\omega_1m|} 
         \left|\Im\left \{ \overline{\omega_1m}\, \omega_2n \right\}\right|
         =
         \frac{1}{|\omega_1m|} 
         \abs{ mn } S 
         \geq
         \frac{S}{\max_i |\omega_i|}\,|n|
   \end{eqnarray*}
   when $m\ne0$, and 
   \begin{eqnarray*}
        |\Omega_{m,n}|^2
         &\geq &
         \frac{S^2}{\max_i |\omega_i|^2}\,n^2
   \end{eqnarray*} 
   for all $m,n\in\bZ$. Hence
   \begin{eqnarray*}
        1-A
         &\leq &  
         \sum_{(m,n)\ne (0,0)}  
         \exp\left[ -\frac{S^2}{4\max_i |\omega_i|^2}\,(m^2+n^2)  \right]
         \\
        &=&
         -1+
         \left(
          \sum_{k \in \bZ}  
          \exp\left[ -\frac{S^2}{4\max_i |\omega_i|^2}\,k^2\right]
         \right)^2\\
        &\leq &
         -1+
         \left(
          1+
          \int_{-\infty}^\infty
          \exp\left[ -\frac{S^2}{4\max_i |\omega_i|^2}\,x^2\right]
          dx
         \right)^2\\
        &=&
         -1+
         \left(
          1+
          \frac{2\sqrt{\pi}\max_i|\omega_i|}{S}
         \right)^2.
   \end{eqnarray*}
   Since $S=|\omega_1|\, |\omega_2|\, \sin \arg\left(\overline{\omega_1}\,\omega_2\right)$, the desired estimate follows.
  \end{proof}

  Theorem~\ref{prop2} proves a part of Theorem~\ref{thm-vnl} directly. 
  Theorem~\ref{prop2} justifies the intuition that, as the lattice becomes large, 
  the von~Neumann lattice approaches an orthogonal family.

  It should be noted however that the uniform success probability cannot be estimated by $S$ only.
  In particular, the condition $S\to \infty$ is not sufficient for the uniform success probability $q(\Pi)$ to approach unity. 
  It can be seen easily as follows.
  The uniform success probability $q(\Pi)$ cannot be greater than the maximal uniform success probability 
  of distinguishing two states $\{ \ket{0}, \ket{\Omega} \}$ in $\vnl(\omega_1,\omega_2)$, 
  where $\Omega$ is either $\omega_1$ or $\omega_2$. 
  One therefore has, by (\ref{173747_27Jul15}), 
  \begin{eqnarray}
   q(\Pi)  \leq 1- |\braket{0|\Omega}| =
        1-\exp\left[-\frac{|\Omega|^2}{2}\right] 
  \end{eqnarray}
  for $\Om=\om_1,\om_2$. 
  This proves that $q(\Pi)$ cannot be estimated solely by $S$.

 \section{Conclusion}

 We  examined distinguishability of countable pure states.
 We defined distinguishability of countable states as the possibility
 of unambiguous measurements on these states. 
 There we classified the distinguishability into three,
 namely, distinguishability, uniform distinguishability, and perfect distinguishability.
 Distinguishability and uniform distinguishability, 
 which are 
 equivalent when the number of states is finite, 
 split when the number becomes infinite.
 We then proved a criterion of distinguishability for countable pure
 states in Theorem~\ref{thm1}. 
 The theorem establishes a relation 
 between operational definitions of distinguishability and intrinsic
 properties of a family of state vectors in the Hilbert space.
 In addition, we gave the maximal uniform success
 probability  
 and a POVM which attain it 
 in Theorem~\ref{prop1}. 

 After developing a general criterion of distinguishability, 
 we discussed distinguishability of von~Neumann's lattice, 
 which is a family of states corresponds to a lattice of the
 phase space in the classical mechanics. 
 Besides its own interest in measurement theory, 
 it serves as an excellent 
 example for the general theory in the sense that 
 all the subtleties arise and can be discussed completely. 
 We showed in  Theorem~\ref{thm-vnl} 
 that distinguishability of a von~Neumann lattice depends only on
 the area of its fundamental region. 
 We see a drastic change in distinguishability 
 when the area becomes the Planck constant $h$.
 The result is robust to deformation of the lattice.
 
 The Planck constant $h$ is without doubt the most fundamental
 constant in quantum physics. 
 It appears in the canonical commutation relation $[Q,P]=ih/(2\pi)$ 
 and characterizes quantum physics in many ways. 
 One is the uncertainty relation of physical quantities. 
 A simple, well-known inequality is Kennard's inequality~\cite[(27)]{kennard1927},
 which gives a bound for 
 the standard deviations $\sigma$ of the observables
 $Q$ and $P$, 
 \[
 \sigma(Q)\sigma(P)\geq \frac{h}{4\pi}. 
 \]
 Modern interpretations of Heisenberg's noise-disturbance uncertainty
 relation and corresponding rigorous inequalities are also found in 
 \cite{PhysRevA.67.042105} (See also \cite[Section~6]{Ozawa2004350}.)
 and \cite{PhysRevA.84.042121}. 
 Another important aspect of $h$ is 
 the Bohr-Sommerfeld quantum condition~\cite[(48.2)]{landau1977quantum}
 \[
 \oint p \, dq=h\left( n+ \frac{1}{2} \right)
 \qquad n=0,1,2,\dots \ .
 \]
 It not only stimulated the discovery of quantum mechanics but also
 can be shown in quantum mechanics with the Wentzel-Kramers-Brillouin (WKB) approximation. 
 This condition explains the familiar fact in statistical mechanics 
 that a single quantum state occupies an area of $h$ in the classical phase.
 Our analysis 
 of von~Neumann lattices revealed 
 the significance of the Planck constant $h$ 
 through the context of state discrimination, 
 thereby giving another measurement-theoretic meaning of $h$, 
 and gives a rigorous version of justification for $h$ 
 to be the unit of the phase space.

 \section{Discussion}

 We have not discussed criterions of distinguishability for mixed
 states. 
 For finitely many mixed states, Feng et al.~\cite{PhysRevA.70.012308}
 obtained a condition of distinguishability.  
 We generalize their result 
 to the case of countably many states
 and present it in a slightly different manner.  
 %
 %
 \begin{prop}
  Let $(\rho_i)_{i\in I}$ be a 
  countable 
  family of 
  density operators on a Hilbert space $H$.
  Then, the following are equivalent.
  \begin{enumerate}
   \item
                $(\rho_i)_{i\in I}$ is distinguishable.
   \item 
                 For all $i\in I$,
                 \begin{eqnarray*}
                  \left(\bigcap_{k\ne i} \ker\rho_k \right) 
                   \setminus \ker \rho_i \ne \varnothing.
                 \end{eqnarray*}
   \item 
                 For all $i\in I$,
                 \begin{eqnarray*}
                  \left(\bigcap_{k \in I} \ker\rho_k \right)
                   \subsetneq  \left(\bigcap_{k\ne i} \ker\rho_k \right).
                 \end{eqnarray*}
   \item
                For all $i\in I$,
                \begin{eqnarray*}
                 \cspan\left( \bigcup_{k \ne i} \overline{\im \rho_k} \right)
                  \subsetneq  \cspan\left( \bigcup_{k\in I} \overline{\im \rho_k} \right).
                \end{eqnarray*}
  \end{enumerate}
  Here, $\overline{ L}$ denotes the norm closure of $L\subset H$,
  $\ker \rho=\set{ \xi \in H| \rho\xi=0}$ and $\im \rho = \set{ \rho\xi | \xi \in H}$
  for a  bounded operator $\rho$ on $H$.
 \end{prop}
 \begin{proof}
  Assume (i). 
  Then, 
  for distinct elements $i$ and $j$ of $I$, 
  one has 
  $0=\tr[\rho_i\Pi_j]=\tr[\Pi_j^{1/2}\rho_i\Pi_j^{1/2}]$ hence 
  $\rho_i\Pi_j^{1/2}=0$. 
  Therefore, 
  for each $i\in I$, there exists $\psi_i$ such that
  $\rho_k(\Pi_i^{1/2}\psi_i) = 0$ for all $k\ne i$ and
  $\rho_i(\Pi_i^{1/2}\psi_i) \ne 0$. 
  This ensures (ii).
  That (ii) implies (i) is as in the proof of
  Theorem~\ref{thm1}. 
  That (ii) is equivalent to (iii) is seen 
  by a trivial set-theoretical identity
  $X\setminus (X\cap Y)=X\setminus Y$. 
  The equivalence of (iii) and (iv) is due to $\overline{\im
  \rho_i}=(\ker \rho_i)^\perp$ and 
$
   \cspan\left( \bigcup_j K_j^\perp \right)
        =
        \left( \bigcap_j K_j \right)^\perp,
$
  where each $K_j$ is a closed subspace of $H$, 
  and $\perp$ denotes the
  orthogonal complement. 
  \end{proof}

  When $(\rho_i)_{i\in I}$ is finite family, (iv) reduces to
  \begin{eqnarray*}
   \lspan\left( \bigcup_{k \ne i} \im \rho_k \right)
        \subsetneq  \lspan\left( \bigcup_{k\in I} \im \rho_k \right) 
  \end{eqnarray*}
  for all $i\in I$.
  This is the condition that Feng et al.~\cite{PhysRevA.70.012308} 
  presented. 
  
  The criterion enables us to investigate the time evolution of distinguishability
  i.e. 
  the relation between distinguishability of 
  $(\rho_i)_{i\in I}$ and that of $(\rho_i')_{i\in I}$.
  Here we assume each state $\rho_i$, at time $t$, evolves to $\rho_i'$ at $t'>t$. 
  We already note in the remark below the Definition~\ref{def-dist} that 
  distinguishability is invariant under unitary evolutions. 
  Thus we should concern non-unitary evolutions of the system,
  which changes a pure state into a mixed state in general.
  Hence we need a criterion for mixed states.

  Distinction of coherent states is not only a subject of theoretical
  interest 
  but also of a practical problem. 
  A coherent state of light is easy to handle and is often
  used  in optical communication. 
  Let us consider the following simple example. 
  The sender generates several coherent states and sends one of them, 
  which travels through optical fibers. 
  The receiver detects it and determine which coherent state was sent.
  The simplest case is to distinguish two states, the vacuum 
  $\ket{0}$ and another coherent state $\ket{\omega_1}$. 
  A slightly more complicated problem is to distinguish nine states, 
  i.e., 
  the vacuum and the eight states enclosing the vacuum 
  $C_9=\{ \, \ket{\Omega} \mid \Omega=0,\pm\omega_1,\pm\omega_2, \pm\omega_1\pm\omega_2 \}$.
  A still more complicated one is distinction of 25 states in the
  set $C_{25}$, doubly surrounding the vacuum. 
  In this manner, we consider distinction of the states in the set
  $C_{1+4n(n-1)}$ which is 
  a finite subset of a von~Neumann lattice. 
  It approaches the whole von~Neumann lattice as $n\to \infty$.
  We denote by $S$ the area of fundamental region of lattice
  corresponding to $C_{1+4n(n-1)}$ (we assume $S>0$). 
  The family 
  $C_{1+4n(n-1)}$ is linearly independent so that by Theorem~\ref{thm1} 
  it is uniformly 
  distinguishable with uniform success probability $q_n>0$. 
  The $q_n$ satisfies $q_1\geq q_2\geq\cdots>0$. 
  Thus there exists a finite limit $\lim_{n\to\infty}
  q_n\geq q$. 
  Here, 
  $q$ is the uniform success probability for the whole lattice 
  which is positive when $S>\pi$ and vanishes otherwise. 
  The behavior of $q_n$ for smaller $n$ may be  of practical interest, 
  and the asymptotic behavior for large $n$ may be of theoretical
  interest.

  \ack
  This work was supported by JSPS Grant-in-Aid for Scientific Research No. 24540282 and by MEXT-Supported Program for the Strategic Research Foundation at Private Universities “Topological Science”.

  \appendix

 \section{Countability of state family}
 \label{counta}
 In this appendix, we briefly discuss the case that states are not countable.
 We shall show that uncountably many states cannot be distinguishable when $H$ is separable. 
 Note that in the standard formulation of quantum mechanics
 (e.g.~\cite[II-1, Postulate E]{von1932mathematische}) the Hilbert space is assumed to be separable.

 We begin by extending the definition of distinguishability to the case of uncountable states. 
 We cannot define POVM with $\sum_{j\in J} \Pi_j=1$ 
 when the set $J$ of outcomes is uncountable since the sum exceeds
 countable additivity.  
 We therefore go back to the measure-theoretical definition of the POVM. 
 Let $B(H)$ be the set of bounded operators on $H$.

 Let $(J,\mathcal J)$ denote a measurable space, where $J$ is a set 
 and $\mathcal J$ a $\sigma$-algebra on $J$. 
 The map $\Pi:\mathcal{J}\to B(H)$ is called a POVM 
 on $(J,\mathcal J)$ 
 if it satisfies the following conditions~\cite[\S3.1]{davies1976quantum}. 
 \begin{enumerate}
  \item  
                 $\Pi(J)=1$ and $\Pi(\Delta)\geq 0$ for all $\Delta\in\mathcal{J}$. 
  \item
           $\Pi\left(\bigcup_k \Delta_k\right)=\sum_k \Pi(\Delta_k)$ 
           in the sense of 
           weak operator topology
           for all disjoint countable collection
           $\set{\Delta_k}\subset \mathcal{J}$. 
 \end{enumerate}

 We shall redefine distinguishability of states, or extend
 Definition~\ref{def-dist}\ref{i-def-dist}, 
 to include families of uncountable states. 
 %
 %
 \begin{defi}\label{defidistunco}
  A POVM $\Pi$ on $(J,\mathcal{J})$ {\em distinguishes}\/ 
  the family $(\rho_i)_{i\in I}$ of states if the following conditions hold. 
  \begin{enumerate}
   \item
                $J=I\sqcup \set{?}$, and $\{i\} \in \mathcal{J}$ holds for all $ i\in I$.
   \item    
                        For all $i,j\in I$, $\tr[\Pi(\set{i})\rho_j]$ is positive if $i=j$ and vanishes if $i\ne j$. 
  \end{enumerate}
  The states $(\rho_i)_{i\in I}$ are {\em distinguishable}\/ 
  if there exists a POVM $\Pi$ which distinguishes them. 
 \end{defi}

 We show a general relation between $\dim H$ and the
 number of input states
 that is necessary for distinguishability. 
 Here, $\dim H$ is defined as the cardinality of an orthonormal basis of $H$, 
 which is countable if and only if $H$ is separable. 
 %
 %
 \begin{thm}
   \label{thm-card}
  Let $(\rho_i)_{i\in I} $ be a state family of a Hilbert space $H$.
  If $(\rho_i)_{i\in I} $ is distinguishable in the sense of Definition
  \ref{defidistunco},  
  then $\dim H \geq |I|$, where $|I|$ denote the cardinality of the set $I$. 
 \end{thm}
 \begin{proof}
  The proof is divided into two steps. 

  1 (Case that all $\rho_i$ are pure).
  Let $\rho_i=\ketbra{\psi_i}, \ \psi_i\in H$.
  Then $(\psi_i)_{i\in I}$ is minimal, which can be shown in a similar
  manner to the proof of Theorem~1. 
  For $J\subset I$, let $K_J:=\cspan\brac{\psi_i\,|\,i\in J}$. 
  Minimality of $(\psi_i)_{i\in I}$ implies that for $k\notin J$, 
  there is a normalized vector $e_k$ which is orthogonal to $K_J$ and
  $\psi_k\in \C e_k+K_J$. 
  Using transfinite induction on $I$, we can construct a orthonormal
  family $(e_i)_{i\in I}$.
  Therefore, $\dim H\geq |I|$.

  2 (General case).
  Assume a POVM $\Pi$ distinguishes $(\rho_i)_{i\in I}$.
  Let $\Pi_j$ denote $\Pi(\set{j})$.
  As in the proof of 
  Proposition~\ref{prop-mixed}, 
  $\Pi_j^{1/2}\rho_i=0$ 
  holds 
  for all $i,j\in I$ 
  with 
  $i\ne j$ 
  and $\Pi_i^{1/2}\rho_i\ne 0$ for all $i\in I$.
  Hence, for each $i\in I$, there exists $ \phi_i\in H$ such that $\Pi_i^{1/2}\rho_i \phi_i\ne0$.
  Define $\psi_i=\rho_i\phi_i/\|\rho_i\phi_i\|$.
  Since 
  \begin{eqnarray*}
   \tr\left[ \Pi_j \ketbra{\psi_i} \right]
        &=\frac{1}{\|\rho_i\psi_i\|^2} \, \left\| \Pi_j^{1/2}\rho_i\phi_i\right \|^2
  \end{eqnarray*}
  for all $i,j\in I$,
  $\Pi$ distinguishes $(\ketbra{\psi_i})_{i\in I} $.
  Therefore, the claim follows from the first step.
 \end{proof}

 Theorem~\ref{thm-card} 
 shows that 
 separability of $H$ forces any distinguishable family of 
 states to be countable.
 On the other hand, 
 there exist distinguishable family of uncountable states in a 
 non-separable Hilbert space.
 A simple example is 
 the following. 
  
 \begin{exam}
  Let $H$ be a non-separable Hilbert space,
  and $(e_i)_{i\in I}$ be 
  a complete orthonormal system of $H$,
  which is uncountable.
  Let $(I,2^I)$ be a measurable space,  
  where $2^I:=\set{ \Delta | \Delta \subset I}$. 
  Define $\Pi:2^J\longto B(H)$ by 
  \begin{eqnarray*}
   \Pi(\Delta)&:=&\sum_{i\in \Delta} \ketbra{e_i},\qquad
        \Delta \in 2^I.
  \end{eqnarray*}
  The sum converges in the strong operator topology. 
  Then $\Pi$ is a well-defined POVM on $ (I, 2^I) $ and 
  (perfectly) distinguishes $(e_i)_{i\in I} $. 
 \end{exam}
 In the example above, 
 $\Pi$ satisfies ``uncountable additivity'' $\sum_{j\in J}
 \Pi(\set{j})=1=\Pi(J)$. 
 However, when we define a measure $\mu_\psi:2^J \longto
 [0,\infty)\subset \bR$  
 for a vector $\psi\in H$ as
 $\mu_\psi(\Delta)=\braket{\psi,\Pi(\Delta)\psi}$,
 then ``uncountable additivity'' of $\mu_\psi$ reduces to a trivial matter
 because $\sum \braket{\psi,\Pi(\set{i})\psi}=\|\psi\|^2<\infty$ and $\set{ j\in J | \mu_\psi(\set{j}) \ne 0 }$ is countable.

 \section{Proof of Proposition~\ref{lembi}}
 \label{proof-lembi}
 We give a proof of Proposition~\ref{lembi} for convenience of the reader.
 We do so by showing a slightly generalized proposition below. 
 For a normed space $X$, let $X^*$ denote the topological dual of $X$, which consists of continuous functionals. 
 In this case, we say $(\psi_i)_{i\in I} $ of $X$ and $(\phi_i)_{i\in I} $ of $X^*$ are biorthogonal 
 when $\phi_i\psi_j=\delta_{i,j}$ for all $i,j\in I$.
 %
 %
 \begin{prop}
   \label{prop-mixed}
  Let $X$ be a normed space and 
  $(\psi_i)_{i\in I}$ be a family of vectors in $X$.
  $(\psi_i)_{i\in I} $ is minimal if and only if $(\psi_i)_{i\in I} $ admits a biorthogonal family.
  If the condition holds and $\cspan\set{\psi_i|i\in I} =X$, 
  then the biorthogonal family is unique.
 \end{prop}

 \begin{proof}
  When $X=0$, the statement is trivial.
  We assume $X\ne0$ in the following.
  Let $Y=\lspan\set{\psi_i | i\in I},\, Y_i=\lspan\set{\psi_j | i\ne j \in I}$
  and these norm closure $\overline{Y},\,\overline{Y_i}$, respectably.

  First, suppose $(\psi_i)_{i\in I}$ is minimal.
  Because $(\psi_i)_{i\in I}$ is linearly independent,
  one can define 
  for each $i\in I$ 
  a linear functional $\phi'_i:Y\longto\mathbb{C}$ 
  by the condition
  $\phi'_i \psi_j=\delta_{i,j}$, where $i,j\in I$.
  We shall show $\phi'_i$ is continuous on $Y$.
  One has 
  \begin{eqnarray*}
    \|\phi'_i\|
    :=\sup_{ \psi\in Y \setminus \{0\} }
    \frac{|\phi'_i\psi|}{\|\psi\|} 
    =\frac{1}{\inf_{\psi\in \psi_i+Y_i} \|\psi\| }. 
  \end{eqnarray*}
  The denominator is positive 
  because $(\psi_i)_{i\in I} $ is minimal.
  Therefore,  
  $\|\phi'_i\|<\infty$ and 
  the linear functional $\phi'_i$ 
  is continuous on $Y$.
  Due to the Hahn-Banach theorem, $\phi'_i$ admits a continuous
  extension to the whole space $X$.
  The extension, 
  which we denote by $\phi_i$, 
  belongs to $X^*$. 
  By construction, 
  $(\phi_i)_{i\in I}$ is a biorthogonal family for $(\psi_i)_{i\in I} $.

  Second, let $(\phi_i)_{i\in I}$ be a biorthogonal family for
  $(\psi_i)_{i\in I}$. 
  Suppose 
  that $(\psi_i)_{i\in I}$ were not minimal.
  Then there would exist $i\in I$ such that $\psi_i \in
  \overline{Y_i}$,  i.e.,  
  there would exist a sequence 
  $(\xi_n)_{n\in\N}$ 
  in $Y_i$ such
  that $\xi_n \to \psi_i$. 
  Since $\phi_i\in X^*$ is continuous and $\phi_iY_i=0$, 
  one would have 
  $
  1
  =\phi_i\psi_i
  =\phi_i\lim_n \xi_n
  =\lim_n \phi_i \xi_n
  =\lim_n 0 
  =0,
  $
  a contradiction.

  For the proof of the uniqueness part of the proposition,
  assume $\overline{Y}=X$ and 
  $(\psi_i)_{i\in I}$ has 
  two biorthogonal families
  $(\phi_i)_{i\in I}$ and $(\phi_i')_{i\in I}$. 
  Because the functional $\phi_i-\phi_i'$ 
  is continuous on $X$ and vanishes on $Y$, 
  one has $\phi_i=\phi_i'$  on $\overline{Y}=X$. 
 \end{proof}

 When $X=H$ is a Hilbert space,
 the Riesz theorem establishes a conjugate isomorphism $H^*\simeq H$ 
 so that 
 a biorthogonal family $(\phi_i)_{i\in I}$ of $H^*$ can be regarded as a family in $H$.


 \end{document}